\documentclass[aps,pra,10pt,notitlepage,a4paper,twocolumn]{revtex4-2}
\usepackage{amsfonts,amssymb,amsmath}            % for math symbols.
\usepackage{lmodern}
\usepackage[]{graphics,graphicx}            % for graphics figures.
\usepackage{amsthm}
\usepackage{bbold,enumitem,mathtools}
\usepackage[skins,breakable]{tcolorbox}
\usepackage{physics}
\usepackage{tikz}
\usetikzlibrary{decorations.pathmorphing,patterns,decorations.markings,matrix,quantikz2,patterns.meta}
\usepackage{adjustbox}
\usepackage{xstring}
\usepackage{ifthen}
\usepackage{listings}% http://ctan.org/pkg/listings
\usepackage{booktabs}
    \usepackage{hyperref}
\usepackage[capitalize]{cleveref}

\usepackage[english]{babel}

\makeatletter
\def\bbl@set@language#1{%
  \edef\languagename{%
    \ifnum\escapechar=\expandafter`\string#1\@empty
    \else\string#1\@empty\fi}%
  \@ifundefined{babel@language@alias@\languagename}{}{%
    \edef\languagename{\@nameuse{babel@language@alias@\languagename}}%
  }%
  \select@language{\languagename}%
  \expandafter\ifx\csname date\languagename\endcsname\relax\else
    \if@filesw
      \protected@write\@auxout{}{\string\select@language{\languagename}}%
      \bbl@for\bbl@tempa\BabelContentsFiles{%
        \addtocontents{\bbl@tempa}{\xstring\select@language{\languagename}}}%
      \bbl@usehooks{write}{}%
    \fi
  \fi}
\newcommand{\DeclareLanguageAlias}[2]{%
  \global\@namedef{babel@language@alias@#1}{#2}%
}
\makeatother

\DeclareLanguageAlias{en}{english}

\newtheorem{conjecture}{Conjecture}

\newtheorem{lemma}{Lemma}

\newtcolorbox[auto counter]{example}[1][]{enhanced,fonttitle=\sffamily\bfseries\large,valign=center
,title={Example \thetcbcounter},label=#1,left=2pt,right=2pt}
\crefname{tcb@cnt@example}{Ex.}{Exs.}

\newenvironment{proof*}[1][\proofname]{%
  
  \begin{proof}[#1]}{\end{proof}}

\newcommand{\identity}{\mathbb{1}}
\renewcommand{\epsilon}{\varepsilon}

\pgfset{
  foreach/parallel foreach/.style args={#1in#2via#3}{evaluate=#3 as #1 using {{#2}[#3-1]}},
}

\newcounter{nodenumber}

\newcounter{arraycard}

\begin{document}

\title{Matrix Inversion by Quantum Walk}
\date{\today}
\author{Alastair \surname{Kay}}\email{alastair.kay@rhul.ac.uk}
\affiliation{Royal Holloway University of London, Egham, Surrey, TW20 0EX, UK}

\author{Christino Tamon}
\email{ctamon@clarkson.edu}
\affiliation{Department of Computer Science, Clarkson University, Potsdam, New York, USA 13699-5815.}
\begin{abstract}
The HHL algorithm for matrix inversion is a landmark algorithm in quantum computation. Its ability to produce a state $\ket{x}$ that is the solution of $Ax=b$, given the input state $\ket{b}$, is envisaged to have diverse applications. In this paper, we substantially simplify the algorithm, originally formed of a complex sequence of phase estimations, amplitude amplifications and Hamiltonian simulations, by replacing the phase estimations with a continuous time quantum walk. The key technique is the use of weak couplings to access the matrix inversion embedded in perturbation theory.
\end{abstract}
\maketitle

The Harrow-Hassidim-Lloyd (HHL) algorithm \cite{harrow2009} is one of the most celebrated results in quantum computing, marking a pivotal step toward quantum algorithms that outperform their classical counterparts. Intended to deliver a solution for systems of linear equations exponentially faster (subject to certain caveats), HHL is a fundamental element of quantum machine learning \cite{rebentrost2014}, solving differential equations \cite{leyton2008,berry2014,montanaro2016}, quantum data fitting \cite{wiebe2012}, and quantum chemistry \cite{baskaran2023}, and has also inspired exciting dequantization strategies \cite{chia2020} providing quantum-inspired classical strategies with improved running times.

Despite its vast potential, the original HHL algorithm is difficult to implement, comprising a complex sequence of phase estimations, amplitude amplifications and Hamiltonian simulations that are largely prohibitive for implementation on near-term, very limited, devices (aside from token attempts \cite{cai2013}). Improvements in asymptotic performance \cite{ambainis2012,childs2017} only exacerbate this. In this work, we present a streamlined variant, eliminating many of these costly components and simplifying the error analysis. Our new approach relies solely on Hamiltonian simulation, generating a quantum walk \cite{farhi1998,farhi2008,ambainis2007}. These have been invaluable in finding speed-ups in various contexts such as the glued-trees algorithm \cite{childs2003}, including experimental implementation \cite{shi2020}, and spatial search  \cite{childs2004}. Structural modifications of the simulated Hamiltonian directly control the error parameters.

For a (possibly non-square) complex matrix $A$, the HHL algorithm \cite{harrow2009} aims to tackle the problem of solving
$
Ax=b.
$
Such a proposition is impractical for a quantum computer as, even if one could prepare a state $\ket{x}$, reading it out would require unreasonably long. Instead, it focusses simply on producing the normalised state $\ket{x}$, having been provided with the state $\ket{b}$. The algorithm is founded upon certain key assumptions, namely, the ability to prepare the initial state $\ket{b}$ and the ability to efficiently implement a Hamiltonian evolution \cite{aaronson2015}. It also requires that the range of singular values of $A$ is bounded between (non-zero) $\frac{1}{\kappa}$ and $1$, where $\kappa$ is known as the condition number of the matrix $A$. We rely on the exact same assumptions. If we allow $A$ to be any arbitrary matrix, \cite{harrow2009} showed that the problem of creating $\ket{x}$ is BQP-hard: at least as hard as any problem that is efficiently solvable on a quantum computer. However, the algorithm made the further assumption that $A$ was sparse. This case has since been dequantized \cite{chia2020}, giving an efficient classical variant provided assumptions about data access can be fulfilled. While the HHL algorithm readily generalises to allow the evaluation of an observable of any state $\ket{x}=f(A)\ket{b}$, we focus entirely on the primary use case of $f(y)=y^{-1}$. Extensions, such as via linear combinations of unitaries \cite{childs2012} or the techniques from the quantum singular value transformation \cite{gilyen2019}, will be possible, but we have no new development. 

Our main insight is that by taking a matrix $H$ into which $A$ is embedded with the promise that all the eigenvalues of $H$ are at least $\frac{1}{\kappa}$ away from 0, then if we couple this using a perturbation $V$ to a new system of close to 0 energy,  i.e.\ $H\rightarrow H\oplus H_0+V$ where $\|V\|\ll\frac{1}{\kappa}$, then we must analyse $H_0$ under degenerate perturbation theory, finding that gaps open up relating to matrix elements of $H^{-1}$. Thus, $H^{-1}$ determines the evolution of states initially supported on the $H_0$ subsystem. The benefit of our approach is in the remarkable simplicity of implementation, removing the need for many repeated iterations of phase estimation. The basic version that we present first matches the asymptotic performance of the original version of the HHL algorithm. By introducing (mild) additional complexity in \cref{sec:better_scaling}, we can match more recent improvements in terms of the scaling with the condition number and the accuracy \cite{ambainis2012,childs2017}.

\section{Main System}

Imagine that we are given an $A$ which we need to invert. This has a singular value decomposition $\{\lambda_n,\ket{\lambda_n},\ket{\eta_n}\}$, i.e.\ $A\ket{\lambda_n}=\lambda_n\ket{\eta_n}$, and we are promised that $A$ has a condition number of $\kappa$, meaning that $\lambda_n\kappa>1$ for all $n$. We assume that $A$ is sparse, which will allow us to invoke a Hamiltonian simulation algorithm such as \cite{berry2015a}.

Our aim is to solve the linear problem $Ax=b$: given a state $\ket{b}$, we produce a normalised version of the state
\begin{equation}\label{eq:target}
\ket{x}=\sum_n\frac{1}{\lambda_n}\ket{\lambda_n}\braket{\eta_n}{b}.
\end{equation}
Since $A$ may not be square, let alone Hermitian, it is necessary to embed it into a larger matrix for implementation as a Hamiltonian. We choose
$$
H=\begin{bmatrix}
0 & \gamma\identity & 0 & 0 \\
\gamma\identity & 0 & A & 0 \\
0 & A^\dagger & 0 & \gamma\identity \\
0 & 0 & \gamma\identity & 0
\end{bmatrix}
$$
where $\gamma$ is a small parameter. $H$ is readily implemented using simulation techniques such as \cite{berry2015a} since its sparsity of only 1 greater than that of $A$. 

The plan is to initialise a state $\begin{bmatrix} \bra{b} & 0 & 0 & 0 \end{bmatrix}^\dagger=\ket{1}\otimes\ket{b}$, evolve it under $H$ for a suitable time $t$, and post-select on a state having arrived in the final subsystem,
$$
\ket{\psi}=%\begin{bmatrix} 0 & 0 & 0 & \identity \end{bmatrix}e^{-iHt}\begin{bmatrix} \ket{b} \\ 0 \\ 0 \\ 0 \end{bmatrix}
(\bra{4}\otimes\identity)e^{-iHt}(\ket{1}\otimes\ket{b}).
$$
% which we denote by the shorthand
% $$
% \ket{\psi}=(\bra{4}\otimes\identity)e^{-iHt}(\ket{1}\otimes\ket{b}).
% $$

\subsection{Analysis}

In order to analyse the evolution of $H$, we decompose the Hamiltonian into a series of 4-dimensional subsystems spanned by the states
\begin{align*}
\ket{n_{1\pm}}&=\frac{1}{\sqrt{2}}\left(\ket{1}\ket{\eta_n}\pm\ket{4}\ket{\lambda_n}\right) \\
\ket{n_{2\pm}}&=\frac{1}{\sqrt{2}}\left(\ket{2}\ket{\eta_n}\pm\ket{3}\ket{\lambda_n}\right).
\end{align*}
The action of $H$ on these states is closed, and is represented by a $4\times 4$ matrix
$$
h_n=\begin{bmatrix}
0 & \gamma & 0 & 0 \\
\gamma & \lambda_n & 0 & 0 \\
0 & 0 & 0 & \gamma \\
0 & 0 & \gamma & -\lambda_n
\end{bmatrix}=h_{n+}\oplus h_{n-}.
$$
We can now express our final evolution in terms of evolution under the different $h_n$:
\begin{align}\label{eq:out}
\ket{\psi}&=\sum_n\ket{\lambda_n}\braket{\eta_n}{b}\alpha_n \\
\alpha_n&=\frac12(\bra{n_{1+}}e^{-ih_nt}\ket{n_{1+}}-\bra{n_{1-}}e^{-ih_nt}\ket{n_{1-}}).\nonumber
\end{align}
Comparing \cref{eq:target,eq:out}, our target is simply to choose $\gamma,t$ such that $\alpha_n\propto\frac{1}{\lambda_n}$ to leading order.

If we take $\theta_n=\lambda_n\sqrt{1+\frac{\gamma^2}{\lambda_n^2}}$, then
\begin{multline}\label{eq:exact}
\alpha_n\propto\left(1+\frac{\lambda_n}{\theta_n}\right)\sin\left(\theta_nt\left(1-\frac{\lambda_n}{\theta_n}\right)\right)\\-\left(1-\frac{\lambda_n}{\theta_n}\right)\sin\left(\theta_nt\left(1+\frac{\lambda_n}{\theta_n}\right)\right).
\end{multline}
We choose $\gamma$ such that $\frac{\gamma}{\lambda_n}<\kappa\gamma\ll1$, allowing us to perform expansions for small $\frac{\gamma}{\lambda_n}$. In particular, $1-\frac{\lambda_n}{\theta_n}=\frac{\gamma^2}{2\lambda_n^2}+O(\gamma^4\kappa^4)$. We consider evolution over a time $t\sim\frac{1}{\gamma}$, meaning that we cannot expand terms like $\sin(\theta_nt)$, but simply bound them by 1. However, we can expand
$$
\sin\left(\theta_nt\left(1-\frac{\lambda_n}{\theta_n}\right)\right)\approx \frac{\theta_nt\gamma^2}{\lambda_n^2}\approx\frac{\gamma}{\lambda_n}.
$$
Tracking the approximations more carefully, we have
$$
\alpha_n\propto\frac{\gamma}{\lambda_n}+O\left(\gamma^2\kappa^2\right).
$$

%\subsection{Success Probability}

Of course, the state $\ket{\psi}$ is not properly normalised, indicating that the protocol only succeeds with some (small) probability, and requires many repetitions. The norm of $\ket{\psi}$ is bounded by $\order{\gamma}$. Thus, in the worst case, the measurement only succeeds with probability $\order{\gamma^2}$. The state we get on outcome, $\ket{\psi}$, compared to the (normalised) target state $\ket{x}$ satisfies
$$
\|\ket{x}-\ket{\psi}\|\sim\gamma\kappa^2.
$$
Thus, we select $\gamma=\frac{\delta}{\kappa^2}$ to give us direct control of the accuracy $\delta$, while also ensuring $\gamma\kappa\ll 1$.

\begin{figure}[b]
\begin{adjustbox}{max width=0.45\textwidth}
\begin{tikzpicture}[darkstyle/.style={circle,draw,fill=black,minimum size=1cm},shaded/.style={darkstyle,preaction={fill, white},pattern={Lines[
                  distance=2mm,
                  angle=45,
                  line width=0.7mm
                 ]},
        pattern color=black}]
\draw [line width=0.6mm] (4,-2) node [darkstyle,fill=white] {} -- (4,0) node [darkstyle] {} node [midway,anchor=west] {$\delta$} -- (2,0) node [darkstyle] {} node [midway,anchor=north] {$\sqrt{7}$} -- (0,0) node [darkstyle] {} node [midway,anchor=north] {$\sqrt{12}$} -- (-2,0) node [darkstyle] {} node [midway,anchor=north] {$\sqrt{15}$} -- (-2,2) node [darkstyle] {} node [midway,anchor=east] {$4$} -- (0,2) node [darkstyle] {} node [midway,anchor=south] {$\sqrt{15}$} -- (2,2) node [darkstyle] {} node [midway,anchor=south] {$\sqrt{12}$} -- (4,2) node [darkstyle] {} node [midway,anchor=south] {$\sqrt{7}$} -- (4,4) node [shaded] {} node [midway,anchor=west] {$\delta$};
\draw [line width=0.6mm] (2,0) -- (2,-2) node [shaded] {}  node [midway,anchor=west] {$\delta$} (2,2) -- (2,4) node [darkstyle,fill=white] {} node [midway,anchor=west] {$\delta$};
\draw [line width=0.6mm] (-2,0) -- (-2,-2) node [shaded] {} node [midway,anchor=west] {$\delta$} (-2,2) -- (-2,4) node [darkstyle,fill=white] {}  node [midway,anchor=west] {$\delta$};
\draw [line width=0.6mm] (0,0) -- (0,-2) node [darkstyle,fill=white] {} node [midway,anchor=west] {$\delta$} (0,2) -- (0,4) node [shaded] {} node [midway,anchor=west] {$\delta$};
\end{tikzpicture}
% \begin{tikzpicture}[darkstyle/.style={circle,draw,fill=black,minimum size=1cm}]
% \foreach \x in {0,...,3}%
%     {%
%     \draw [thick] (4*\x,0) node (a\x) [darkstyle] {} -- (4*\x,-1.5) node (b\x) [darkstyle] {} node [pos=0.5,anchor=east] {$\delta$};
%     \draw [thick] (4*\x+2,0) node (c\x) [darkstyle] {} -- (4*\x+2,1.5) node (d\x) [darkstyle] {} node [pos=0.5,anchor=east] {$\delta$};
%     % \foreach \y in {0,...,1}%
%     %     \node [darkstyle]  (\x\y) at (2*\x,1.5*\y) {};
% }
% \draw [thick] (0,0) -- (14,0) node[pos=1/14,anchor=south] {$\sqrt{7}$} node[pos=3/14,anchor=south] {$\sqrt{12}$}  node[pos=5/14,anchor=south] {$\sqrt{15}$} node[pos=7/14,anchor=south] {$4$} node[pos=9/14,anchor=south] {$\sqrt{15}$}  node[pos=11/14,anchor=south] {$\sqrt{12}$} node[pos=13/14,anchor=south] {$\sqrt{7}$};

% \draw [decorate,decoration={brace,amplitude=10pt}]
% ($(d0.north west)+(0,0.3cm)$) -- ($(d3.north east)+(0,0.3cm)$) node [black,midway,anchor=south,yshift=0.5cm]
% {$\ket{b}$};
% \draw [decorate,decoration={brace,amplitude=10pt}]
% ($(b3.south east)+(0,-0.3cm)$) -- ($(b0.south west)+(0,-0.3cm)$) node [black,midway,anchor=north,yshift=-0.5cm]
% {$\ket{x}$};
% \end{tikzpicture}
\end{adjustbox}
  \caption{Network of qubits coupled by an exchange coupling of the indicated strengths and implementing \cref{ex:main} within its single excitation subspace. The state $\ket{b}$ is initialised across the open circles, and $\ket{x}$ is received on the hatched circles.}\label{fig:network}  
\end{figure}
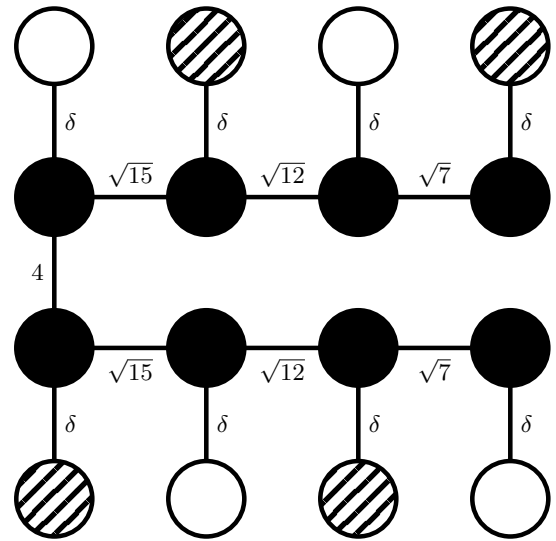

The protocol must be repeated $\order{\frac{\kappa^4}{\delta^2}}$ times. However, as in \cite{harrow2009}, we can replace this measure and repeat protocol with amplitude amplification \cite{grover1998,brassard2002} in order to run for only time $\order{\frac{\kappa^2T}{\delta}}$ where $T$ is the time required for the Hamiltonian simulation. This is the same accuracy, as a function of $\kappa$, as the original version of HHL \cite{harrow2009,ambainis2012}, although this has since been improved \cite{ambainis2012}.

\begin{example}[ex:main]
We wish to solve the problem
$$
\begin{bmatrix}
\sqrt{7} & 0 & 0 & 0 \\
 2 \sqrt{3} & \sqrt{15} & 0 & 0 \\
 0 & 4 & \sqrt{15} & 0 \\
 0 & 0 & 2 \sqrt{3} & \sqrt{7}
\end{bmatrix}x=\begin{bmatrix}
\sqrt{21} \\ 6 \\ 0 \\ \sqrt{35}
\end{bmatrix}.
$$
We set up the Hamiltonian and (normalized) initial state as specified \cite{kay2025b}, evolving the Hamiltonian for a time $\delta^{-1}$ and taking $\delta=0.01$ ($\kappa=1$). Post-selecting on the final block, we get a state
$$
\ket{\psi}=\begin{bmatrix} 0.6117 \\  -0.00076 \\ -0.0008 \\ 0.7911\end{bmatrix}.
$$
The renormalised version of the true answer is
$$
\ket{x}=\frac{1}{\sqrt{8}}\begin{bmatrix}
\sqrt{3} \\ 0 \\ 0 \\ \sqrt{5}
\end{bmatrix}\approx\begin{bmatrix}
0.6124 \\ 0 \\ 0 \\0.7906
\end{bmatrix},
$$
achieving the anticipated $\order{\delta}$ accuracy.

This instance can be implemented using just four qubits, initialised as $\ket{00}\ket{b}$ and arriving as $\ket{11}\ket{x}$ after post-selecting on the first two qubits being $\ket{11}$. All the $\delta$-strength terms are created by a single Pauli term $\delta\identity\otimes X\otimes\identity\otimes\identity$. Alternatively, it maps to the single excitation subspace of a network of spins coupled using the exchange or Heisenberg coupling with a topology as depicted in \cref{fig:network} \cite{kay2010a}. This may be particularly amenable to short-term realisation, either in photonic systems \cite{keil2010,chapman2016} or superconducting ones \cite{li2018}. %This is remarkably similar to a configuration described in \cite{burgarth2006a}, where the main system qubits (the central spins, coupled by $A$ and $A^\dagger$) are each attached to a `bath' of one spin.
\end{example}

\section{Enhanced Performance}\label{sec:better_scaling}

Errors in our protocol can be ascribed to two sources. In \cref{eq:exact}, the first term describes the evolution of a state initially supported on the low energy ($\sim 0$) eigenvectors, giving exactly the term we want at first order, but higher order terms (all odd) mark deviations. The second term describes the (small) probability of being found in a high energy ($\sim\lambda_n$) eigenstate, with unwanted evolution. This second term is the leading order error, and is readily controlled by further extending the Hamiltonian,
\begin{equation}\label{eq:ham_mid}
H=\begin{bmatrix}
0 & \gamma\identity & 0 & 0 & 0 & 0 \\
\gamma\identity & 0 & \gamma\identity & 0 & 0 & 0 \\
0 & \gamma\identity & 0 & A & 0 & 0 \\
0 & 0 & A^\dagger & 0 & \gamma\identity & 0 \\
0 & 0 & 0 & \gamma\identity & 0 & \gamma\identity \\
0 & 0 & 0 & 0 & \gamma\identity & 0
\end{bmatrix}.
\end{equation}
The effective Hamiltonians
$$
h_{n\pm}=\begin{bmatrix}
    0 & \gamma & 0 \\
    \gamma & 0 & \gamma \\
    0 & \gamma & \pm\lambda_n
\end{bmatrix}
$$
have eigenvectors of energy $\sim\lambda_n$ which are supported on state $\ket{n_{1\pm}}$ with a strength reduced to order $\gamma^2$ (\cref{lem:exp}). The low energy eigenvalues are of the form
$$
\gamma\left(\pm 1\pm\frac{\gamma}{2\lambda_n}\right)+\order{\gamma^3}.
$$
We choose to evolve for a time $t=2\pi/\gamma$ to ensure that the leading $\pm 1$ terms vanish. % (this is not strictly necessary, $\order{1/\gamma}$ is sufficient, as even powers of $\gamma$ vanish in the final calculation).
 We now have that
$$
\alpha_n\propto\frac{\gamma}{\lambda_n}+\order{\gamma^3\kappa^3}.
$$
The accuracy of the final normalised state is $\gamma^2\kappa^3$; selecting $\gamma=\delta/\kappa^{3/2}$ improves the run-time to $\order{\kappa^{3/2}T/\delta}$.

Further refinements are possible: by extending $H$ so that $h_{n+}$ is an $(R+1)\times (R+1)$ matrix, the effect of the high energy sector is suppressed to $\order{\gamma^{2R}}$ (\cref{lem:exp}). Furthermore, expansion of $\alpha$ for the low-energy sector retains only odd powers of $x$ (see \cref{lem:odd}), and these coefficients are independent of $\lambda$. Thus, by selecting evolutions over multiple times, we can take a linear combination of unitaries \cite{childs2012} and cancel all errors up to $\order{\gamma^{2R}}$. Accuracy $\epsilon$ only requires $R\sim\log\frac{1}{\epsilon}$.

\begin{example}[ex:mid]
In the specific case of \cref{ex:main}, the Hamiltonian in \cref{eq:ham_mid} gives an evolution with an accuracy $\|\ket{\psi}-\ket{x}\|=2\times 10^{-4}$. This improves to $7\times 10^{-9}$ by choosing a linear combination that eliminates the $\order{\gamma^3}$ term:
\begin{center}
\begin{adjustbox}{max width=\textwidth}
$
\displaystyle\bra{6}\otimes\identity\left(\frac{2 \left(4 \pi ^2-9\right)e^{-iH\frac{\pi}{\gamma}}+(\pi ^2-9)e^{-iH\frac{2\pi}{\gamma}}}{\sqrt{405-306 \pi ^2+65 \pi ^4}}\right)\ket{1}\otimes\ket{b}.
$
\end{adjustbox}
\end{center}

\end{example}

\subsection{High-Order Corrections by Design}

We believe that a more direct method can eliminate the complexities of a linear combination of unitaries.

\begin{conjecture}\label{conjecture}
There exist coupling strengths $\{J_n\}_{n=1}^R$ such that
\begin{multline*}
H=\sum_{n=1}^R\gamma J_n\Big(\ket{n}\bra{n+1}+\ket{n+1}\bra{n}\Big)\otimes\identity_N+\\
\sum_{n=R+2}^{2R+1}\gamma J_{2R+2-n}\Big(\ket{n}\bra{n+1}+\ket{n+1}\bra{n}\Big)\otimes\identity_M\\
+\ket{R+1}\bra{R+2}\otimes A+\ket{R+2}\bra{R+1}\otimes A^\dagger
\end{multline*}
achieves an evolution
$$
(\bra{2R+2}\otimes\identity)e^{-iHt}(\ket{1}\otimes\ket{b})\propto\gamma \ket{x}+\order{\gamma^{2R}\kappa^{2R}}.
$$
%with a running time $\order{\kappa^\frac{2R}{2R-1}T}$ to achieve an accuracy $\epsilon\sim(\gamma\kappa)^{2R}$.
\end{conjecture}
%It is worth noting that $T$ scales roughly as $\frac{s}{\gamma}(\log N+\log R)\log\frac{s}{\epsilon\gamma}$ where $H$ is $s$ sparse \cite{berry2015a}.

\begin{table*}
    \centering
\begin{tabular}{cccccccc}
\toprule
$R$ & $J_1$ & $J_2$ & $J_3$ & $J_4$ & $J_5$ & $J_6$ & accuracy \\\midrule
1 & any &&&&&& $x^2$ \\
2 & 0.5 &  0.494159 &&&&& $x^4$\\
3 & 0.601912 & 0.798563 & 0.632067 &&&& $x^6$ \\
4 & 0.755581 & 0.971502 & 0.992613 & 0.713346 &&& $x^7$ \\
5 & 0.843007 & 1.1131 & 1.25361 & 1.21606 & 0.83284 && $x^9$\\
6 & 0.914772 & 1.20074 & 1.38343 & 1.53701 & 1.4816 & 0.963929 & $x^9$\\
% $R$ & $J_R^2$ & $a_1$ & $a_2$ & $a_3$ & accuracy \\\midrule
% 1 & n/a & 1 &&& $x^2$ \\
% 2 & 0.24419 & $\frac12$&&&$x^4$\\
% 3 & 0.39951 & 0.31885 & 0.36230 && $x^6$ \\
% 4 & 0.50886 &0.18382&0.31618&&$x^7$\\
% 5 & 0.69362&0.12634&0.23406&0.27921&$x^9$\\
% 6 & 0.92916&0.105815&0.16884& 0.22534&$x^9$ \\
\bottomrule
\end{tabular}
\caption{Parameters required for a given length $R$ of block-wise extension so that, with fixed eigenvalues, the evolution yields a leading order term of $x\propto\frac{1}{\lambda}$, with higher orders suppressed at time $t=2\pi/\gamma$. The $R=1$ case is the one we initially solved, and is limited by the effect of the high energy term, as are all the small $R$ cases (accuracy increases by $x^2$ with each increment of $R$), transitioning at $R=4$ (accuracy increases by $x^2$ for every second increment of $R$).}\label{tab:params}
\end{table*}

In \cref{sec:app}, we set up the framework, and a simplifying restriction, for finding the coupling strengths. For $R\leq 6$, the numerical solutions are given in \cref{tab:params}. We do not know whether a valid solution to the equations always exists, let alone the efficiency of its discovery. 

\begin{example}[ex:ultimate]
Take $R=3$. We set up a Hamiltonian
$$
H=\gamma\begin{bmatrix}
0 & J_1\identity & 0 & 0 & 0 & 0 & 0 & 0 \\
J_1\identity & 0 & J_2\identity & 0 & 0 & 0 & 0 & 0 \\
0 & J_2\identity & 0 & J_3\identity & 0 & 0 & 0 & 0 \\
0 & 0 & J_3\identity & 0 & A/\gamma & 0 & 0 & 0 \\
0 & 0 & 0 & A^\dagger/\gamma & 0 & J_3\identity & 0 & 0 \\
0 & 0 & 0 & 0 & J_3\identity & 0 & J_2\identity & 0 \\
0 & 0 & 0 & 0 & 0 & J_2\identity & 0 & J_1\identity \\
0 & 0 & 0 & 0 & 0 & 0 & J_1\identity & 0
\end{bmatrix}
$$
with coupling strengths specified in \cref{tab:params}. Using the basis $\ket{m_+}=(\ket{m}\ket{\eta_n}+\ket{9-m}\ket{\lambda_n})/\sqrt{2}$, $H$ decomposes into subspaces within which the coupling is $h_+$. The leading submatrix of $h_+/\gamma$ is
$$
\tilde h=\begin{bmatrix} 0 & J_1 & 0 \\ J_1 & 0 & J_2 \\ 0 & J_2 & 0 \end{bmatrix},
$$
designed to have eigenvalues $0,\pm 1$. We find
\begin{align*}
\alpha_n%&=\frac12(\bra{1_+}e^{-ih_+2\pi/\gamma}\ket{1_+}-\bra{1_+}e^{ih_+2\pi/\gamma}\ket{1_+})\\
&=0.869923 i\frac{\gamma}{\lambda}+\order{\frac{\gamma^6}{\lambda^6}}.
\end{align*}
We thus have
$$
(\bra{8}\otimes\identity)e^{-iH2\pi/\gamma}(\ket{1}\otimes\ket{b})\propto\ket{x}+\order{\gamma^5\kappa^6}.
$$

Applied to the special case of \cref{ex:main} (see \cite{kay2025b}), the final state $\ket{\psi}$ satisfies $\|\ket{\psi}-\ket{x}\|=2.6\times 10^{-13}$, demonstrating just how effectively this works. In terms of an experimental implementation, we simply update \cref{fig:network} by replacing each $\delta$-strength coupling with a chain of three edges with couplings $\delta\times\{J_1,J_2,J_3\}$ ($\delta J_3$ attaches to the bulk).
\end{example}

Assuming that \cref{conjecture} holds for any $R$, we achieve an overall accuracy $\order{\gamma^{R}\kappa^{R+1}}$. To ensure that the simulation achieves an accuracy of at least $\epsilon$, we require both that the Hamiltonian simulation needs an accuracy of $\epsilon/2$ and the current protocol does as well, i.e.\ $\epsilon\sim \gamma^{R}\kappa^{R+1}$, which requires $R\sim\log\frac{1}{\epsilon}$. As $R$ becomes large, the $\kappa$ dependence tends to the optimal \cite{harrow2009,ambainis2012}, and in terms of accuracy, we recover \cite{childs2017} although without the need for any phase estimation, linear combinations of unitaries etc.; just the simulation of a Hamiltonian evolution followed by amplitude amplification. In early, small-scale, experiments, we anticipate replacing the amplitude amplification with repeat-until-success measurements, which have a lower overhead. Thus, the entire experiment is just a Hamiltonian evolution.

% By taking $t=2\pi/\gamma$, the $\sin()$ function can have a small expansion in $x$, with the leading term giving us the $\frac{1}{\lambda_m}$ functionality that we need. Note that all even powers of $x$ then vanish by virtue of the $\pm\eta^{(0)}$ pairs (to be carefully proven). If we consider the $x^3,x^5$ etc.\ terms, then these are just linear combinations of the $\{a_n\}$. So long as there is a solution requiring positive $a_n$, and also satisfying $\sum_na_n=1$, we can rebuild the desired chain and hence have all these terms, up to some power, cancelling. Presumably it is also possible to get to the $x^{2R}$ terms. This would give us the same performance as \cite{childs2022}!!

\vspace{-0.125in}
\section{Conclusions}

In this paper, we have demonstrated a new version of the algorithm for matrix inversion. Without the need for phase estimation, it is far more direct that HHL. At its core is an implicit perturbative analysis of a weakly coupled system, whose effective Hamiltonian for a subspace close to energy $0$ is approximately $H^{-1}$. The simplest variant of this algorithm matches the $\order{\kappa^2}$ scaling of the original algorithm \cite{harrow2009} while delivering a massive simplification of the implementation of the algorithm, making it vastly more realistic for short-term realisation on Noisy Intermediate Scale Quantum devices. With a mild increase in sophistication, we have also matched the best performance \cite{ambainis2012,childs2017} in terms of accuracy and condition number. However, we have left open the question (\cref{conjecture}) of whether the most direct version always has solutions that can be found efficiently.

We have not investigated functions $\ket{x}\propto f(A)\ket{b}$ other than the inverse in this paper. Standard techniques such as linear combinations of unitaries \cite{childs2013,chia2020} will suffice, although it would be more satisfying if the function could be directly integrated into the evolution.

\section*{Acknowledgments}
C.T. is supported by NSF grant OSI-2427020.

%\newpage
%\bibliography{References}

\clearpage
\appendix
\section{Partial Progress Towards Conjecture}\label{sec:app}

As before, we consider subspaces based on the singular value decomposition of $A$. We get an effective $(R+1)\times (R+1)$ matrix
$$
h_+=\gamma\left[\begin{array}{cccccc}
0 & J_1 & & & & \\
J_1 & 0 & J_2 & & & \\
& J_2 & 0 & J_3 & & \\
& & \ddots & \ddots & \ddots & \\
& & & J_{R-1} & 0 & J_{R} \\
& & & & J_{R} & \frac{1}{x}
\end{array}\right]
$$
where $x=\frac{\gamma}{\lambda}$ is the small parameter that we will be performing a series expansion about close to 0. We will actually concentrate on a reduced case where the spectrum of $\tilde h$, the principal submatrix of $h_+/\gamma$, is fixed ($0,\pm1,\pm 2, \ldots$) \footnote{Although even powers of $x$ in $\alpha(x)$ vanish, it helps computer packages to explicitly zero the $\order{1}$ term by fixing $t=\frac{2\pi}{\gamma}$. This choice of spectrum has implications for $\|\tilde h\|$, which needs renormalising, and will therefore impact the evolution time of the Hamiltonian simulation.}, and consequently we only achieve an accuracy $\order{\gamma^R\kappa^R}$.

\begin{lemma}\label{lem:exp}
$h_+$ has an eigenvector $\ket{\lambda}$ with eigenvalue $\lambda\sim\frac{1}{x}$. There exists a positive constant $\beta$ that bounds the first element of the eigenvector by
$$
|\braket{1}{\lambda}|<\beta |x|^{R}
$$
provided $\|\tilde h\|\gamma\kappa\ll 1$.
\end{lemma}
\begin{proof}
The existence of the eigenvector is justified in perturbation theory by Kato \cite{kato1995a}. We try to solve for the unnormalised vector $\ket{\lambda}=\begin{bmatrix}\ket{\lambda_L} \\ 1 \end{bmatrix}$,
$$
\tilde h\ket{\lambda_L}+\ket{R}=\lambda\ket{\lambda_L}.
$$
We rearrange this as
$$
\braket{1}{\lambda_L}=\bra{1}(\lambda\identity-\tilde h)^{-1}\ket{R}.
$$
If we expand the inverse as a series in small $\frac{\|\tilde h\|}{\lambda}$, we have
$$
\braket{1}{\lambda_L}=\sum_{k=0}^{\infty}\frac{1}{\lambda^{k+1}}\bra{1}\tilde h^k\ket{R}.
$$
For values $k<R-1$, $\bra{1}\tilde h^k\ket{R}=0$ due to the chosen tridiagonal structure. We can thus bound
$$
|\braket{1}{\lambda_L}|\leq \sum_{k=R-1}^{\infty}\frac{1}{\lambda^{k+1}}\|\tilde h\|^k=\frac{1}{\lambda^{R}}\frac{\|\tilde h\|^{R-1}}{1-\frac{\|\tilde h\|}{\lambda}}.
$$
\end{proof}

The larger system size, irrespective of the choice of couplings, suppresses the effects any the high energy terms by an amount $\braket{1}{\lambda_L}^2\sim x^{2R}$.

In this case, our initial form of $h_+$ is not so convenient. Instead, we perform a similarity transform, making use of the fact that a tridiagonal matrix is uniquely specified by its eigenvalues and the weight of the corresponding eigenvectors on a site at one end ($a_n$). Moreover, for our chosen structure, $a_{R-n}=a_n$
$$
h_+\sim\begin{bmatrix}
    \frac{R}{2} & 0 & 0 & \ldots & 0 & 0 & J_R^2a_1 \\
    0 & \frac{R}{2}-1 & 0 & &0 & 0 & J_R^2a_2 \\
    0 & 0 & \ddots&&&&\vdots\\
    \vdots \\
    0 & 0 &&&0 & -\frac{R}{2}  & J_R^2a_1 \\
    1 & 1 & 1 & \ldots & 1 & 1 & \frac{1}{x}
\end{bmatrix}
$$
subject to the constraints
$$
\sum_na_n=1,\qquad a_n>0.
$$
The amplitude that we want to approximate is
$$
\alpha=\frac12\left(\bra{1}e^{-ih_+t}\ket{1}-\bra{1}e^{ih_+t}\ket{1}\right).
$$
Our goal, as before is to have the leading order proportional to $x$, but now to cancel as many higher order terms as possible.

With this parametrisation in place, we can solve for the eigenvalues $\phi_n$ of $h_+$ that are close to the initial eigenvalues of $(0,\pm 1,\pm 2,\ldots)\gamma$ as a series in $x$
$$
\phi_n=\frac{R}{2}+1-n+\sum_{k=1}^{\infty}\beta_{n,k}x^k.
$$
In practice, we only need the $\beta_{n,k}$ up to $k=R$ ($2R$ in the full case), which can be found iteratively.

The structure of the matrix $h_+$ has some immediate implications for the structure of the eigenvalues and eigenvectors. Specifically,
\begin{lemma}
$$
\beta_{n,k}=(-1)^{k+1}\beta_{R+1-n,k}
$$
\end{lemma}
\begin{proof}
Consider $h_+$ as a function of $x$, i.e.\ $h_+(x)$. If we introduce the matrix $D=\sum_{n=1}^{R+1}(-1)^{n+1}\proj{n}$, then
$$
h_+(x)=-Dh_+(-x)D.
$$
The existence of an eigenvector $\ket{\phi(x)}$ of eigenvalue $\phi(x)$ means that there is also an eigenvector $D\ket{\phi(-x)}$ of eigenvalue $-\phi(-x)$. In the eigenvalue, these signs flip the sign of all even powers of $x$. Similarly, $b_n(x)=\braket{R+1}{\phi_n}^2=b_{R+1-n}(-x)$.    
\end{proof}
\noindent Consequently, many of the coefficients of $\alpha$ are 0.
\begin{lemma}\label{lem:odd}
$\alpha$ is an odd function of $x$.
\end{lemma}
\begin{proof}
    Let us write $b_n=C+xD$ and $\eta_n=A+xB$ where $A,B,C,D$ are even functions of $x$. We have $b_{R+1-n}=C-xD$ and $\phi_{R+1-n}=-A+xB$. The coefficient $\alpha$ is calculated by
    \begin{align*}
    \alpha&=\frac12\sum_n b_n\sin(t\phi_n)+b_{R+1-n}\sin(t\phi_{R+1-n}) \\
    %&=\frac12\sum_n (C+Dx)\sin(t(A+Bx))+(C-Dx)\sin(t(-A+Bx)) \\
    &=\sum_n C\sin(x B t)\cos(At)+xD\sin(At)\cos(xBt).
    \end{align*}
    Exactly one of each term in each product is odd, with the others being even. Hence, the overall function is odd.
\end{proof}

Note that given the eigenvalues $\phi_n$, we can calculate the one additional eigenvalue $\frac{1}{x}-\sum_n\phi_n$. This must be an odd function of $x$. We can also evaluate the $b_n$ as required for calculating $\alpha$:
$$
b_n\propto\frac{1}{p'(\eta_n)q(\phi_n)}
$$
where $q(z)$ is the characteristic polynomial of $\tilde h$ (whose eigenvalues we already know, having fixed them at the start) and $p'(z)$ is the derivative of the characteristic polynomial of $h_+$.

If we now find the expansion of $\alpha$ in terms of $x$, then there are terms $x^3$, $x^5$, \ldots $x^{2\lceil R/2\rceil+3}$ that can potentially be zeroed using the $\lceil R/2\rceil+1$ parameters $a_n$ and $J_R$. Having control of the eigenvalues gives another $\lfloor R/2\rfloor$ parameters. By doing this separately for each power of $x$, the identity will work for all possible $\lambda$, the singular values of $A$. In \cref{tab:params}, we list some example values that work for modest system sizes. %It is perhaps useful to note that the pattern is almost symmetric, rising from the extremes and peaking in the middle of the chain. As such, we expect these couplings to be close to the perfect transfer couplings \cite{christandl2004} $J_n\sim\sqrt{n(R+1-n)}$ (perhaps a useful starting point for a numerical iteration?). This also implies that the largest of these grows linearly with $R$. It may be necessary to rescale in order to ensure that the expansions remain valid, i.e.\ increase the evolution time by a factor of $R$. We must factor this into our running time.

\end{document}